\documentclass[amsmath,amssymb,showkeys, showpacs,superscriptaddress]{revtex4}
\usepackage{graphicx}
\DeclareMathAlphabet{\mathpzc}{OT1}{pzc}{m}{it}
\usepackage{latexsym}
\usepackage{amsmath,amsmath,amsfonts, amssymb,MnSymbol}
\usepackage{amsthm}
\usepackage{rotating}
\newtheorem{theorem}{Theorem}[section]

\begin{document}

\vspace*{0.4 in}
\title{On the solvability of confluent Heun equation and associated orthogonal polynomials}
\author{Nasser Saad}
\email{nsaad@upei.ca}
\affiliation{Department of Mathematics and Statistics,
University of Prince Edward Island, 550 University Avenue,
Charlottetown, PEI, Canada C1A 4P3.}

\begin{abstract}
\noindent \textbf{Abstract:} The present paper analyze the constraints on the confluent Heun type-equation, 
$(a_{3,1}r^2+a_{3,2}r)y''+(a_{2,0}r^2+a_{2,1}r+a_{2,2})y'-(\tau_{1,0}r+\tau_{1,1})y=0,$ where  $|a_{3,1}|^2+|a_{3,2}|^2\neq 0,
$ and $a_{i,j},i=3,2,1, j=0,1,2$ are real parameters, to admit polynomial solutions. The necessary and sufficient conditions for the existence of these polynomials are given. A three-term recurrence relation is provided to generate the polynomial solutions explicitly. We, then, prove that these polynomial solutions are a source of finite sequences of orthogonal polynomials. Several properties, such as the recurrence relation, Christoffel-Darboux formulas and the moments of the weight function, are discussed. We also show a factorization property of these orthogonal polynomials that allow for the construction of other sequences of orthogonal polynomials. For illustration, we examines the quasi- exactly solvability of  the $(p,q)$-hyperbolic potential $V(r)=-V_0\sinh^p(r)/\cosh^q(r), V_0>0, p\geq 0, q>p$.  The associated orthogonal polynomials generated by the solutions of the Schr\"odinger equation with the $(4,6)$-hyperbolic potential are constructed.

\end{abstract}

\maketitle
\section{Introduction} 
\noindent The second-order differential equation 
\begin{align}\label{eq1}
(a_{3,1}r^2+a_{3,2}r)\,\frac{d^2y}{dr^2}+(a_{2,0}r^2+a_{2,1}r+a_{2,2})\,\frac{dy}{dr}-(\tau_{1,0}r+\tau_{1,1})y=0,\qquad a_{3,1}^2+a_{3,2}^2\neq 0,~~a_{2,0}^2+\tau_{1,0}^2\neq 0,
\end{align}
with parameters  $a_{i,j}\in \mathbb R,~i=3,2,1,~ j=0,1,2$ has two regular singular points $r=0$ and $r=-a_{3,2}/a_{3,1}$ with exponents $\{0,1-a_{2,2}/a_{3,2}\}$  and $\{0,1-{a_{2,1}}/{a_{3,1}}+{a_{2,2}}/{a_{3,2}}+{a_{2,0}a_{3,2}}/{a_{3,1}^2}\}$,  respectively, in addition to an irregular singular point at infinity unless $a_{2,2}=\tau_{1,1}=0$. 
\vskip0.1true in
\noindent  Although equation \eqref{eq1} generalized the standard confluent Heun equation
\cite{ron1995}, we shall keep the same name `the confluent Heun equation' for \eqref{eq1} as well.
The name confluent driven from the coalescence process of singularities of the general Heun equation \cite{ron1995,dec1978a,dec1978b}. 
\vskip0.1true in
\noindent  Over the past few decades, the confluent Heun equation appeared intensively in mathematical physics applications. Including but not limited to, in studying the spherical Coulomb functions, the spheroidal wave functions, the Mathieu equation and confinement potentials in quantum mechanics. Most recently, it also found some interesting applications in General Relativity and Astrophysics
 \cite{saad2009,fiz2010,jaick2008,cheb2004,fiz2011,chris2011,ru2011}.
\vskip0.1true in
\noindent The purpose of this article is twofold. First, we analyze the necessary and sufficient conditions under which the differential equation \eqref{eq1} admit the polynomial solutions $y_n(r)=\sum_{k=0}^n \mathcal C_k r^k$, $n=0,1,2,\dots$, Where we set up a scheme for constructing the polynomial coefficients $\mathcal C_k$ explicitly. Second, we establish a correspondence between each solution $y_n(r)$ and a finite sequence of orthogonal polynomials 
$\{P_k(\tau_{1,1})\}_{k=0}^n$. We discuss some properties of these orthogonal polynomials such as the recurrence relation,  factorization, Christoffel-Darboux formulas, and the moments of the weight function.
\vskip0.1true in
\noindent  As an illustration of these polynomials, we explore the quasi- exact solutions of the one-dimensional Schr\"odinger equation with a hyperbolic potential \cite{downing2013} 
\begin{align}\label{eq2}
-\frac{\hbar^2}{2m}\dfrac{d^2}{dx^2}\psi_n(x)-V_0\,\frac{\sinh^p(x/d)}{\cosh^q(x/d)}\psi_n(x)=E\,\psi_n(x),\quad -\infty<x<\infty,
\end{align}
with special attention to the case of $p=4$ and $q=6$. We show under certain  transformations, this equation reduces to a confluent Heun equation of type \eqref{eq1}. We discuss the associated sequence of orthogonal polynomials $\{\mathpzc{P}_k\}_{k=0}^n$ generated by the solution $\psi_n,n\geq 1$ and investigate some of their properties.
\vskip0.1true in
\noindent The article organized as follows. In Section II, we gave the necessary and sufficient conditions for the polynomial solvability constraints of equation \eqref{eq1} with a general procedure for formulating the polynomial solution explicitly. In Section III, we establish a correspondence between that the solution $y_n$ and a set of orthogonal polynomials 
$\{\mathpzc{P}_k\}_{k=0}^n$. In that section, we also derive the recurrence relation, the Christoffel-Darboux formula and prove the factorization property of these polynomials. The weight functions and some explicit expressions calculating the moments of the weight functions also discussed.   In Section IV, we discuss the quasi- exact solvable quantum model, given by a hyperbolic potential and explicitly construction the finite sequences of orthogonal polynomials associated with the wavefunction solutions. 
\section{Polynomial solvability constraints}
\noindent The $n$-th derivatives of equation \eqref{eq1} yields
\begin{align}\label{eq3}
(a_{3,1}r^2&+a_{3,2}r)y^{(n+2)}+(a_{2,0}r^2+[a_{2,1}+2na_{3,1}]r+na_{3,2}+a_{2,2})y^{(n+1)}\notag\\
&+\left((2na_{2,0}-\tau_{1,0})r+n(n-1)a_{3,1}+na_{2,1}-\tau_{1,1}
\right)y^{(n)}+\left(n(n-1)a_{2,0}-n\tau_{1,0}
\right)y^{(n-1)}=0,\quad n=0,1,\dots.
\end{align}
Thus, the necessary condition for the $n$-degree polynomial solution is $\tau_{1,0}=n\, a_{2,0}, n=0,1,\dots
$
while the sufficient condition follow using the recurrence relation
\begin{equation}\label{eq4}
(k+1)\left(ka_{3,2}+a_{2,2}\right){\mathcal C}_{k+1}+\left(k(k-1)a_{3,1}+ka_{2,1}-\tau_{1,1}\right)){\mathcal C}_k+\left((k-1)a_{2,0}-\tau_{1,0}\right)){\mathcal C}_{k-1}=0,\quad {\mathcal C}_{-1}=0, \quad {\mathcal C}_{0}=1,
\end{equation}
according to the following simple procedure.
\vskip0.1true in
\noindent  For the zero-degree polynomial solution $y_0(r)= 1$: the necessary and sufficient conditions, respectively, reads 
\begin{equation}\label{eq5}
\tau_{1,0}=0,\qquad\tau_{1,1}=0. 
\end{equation}
For the first-degree polynomial solution 
 \begin{equation}\label{eq6}
y_1(r)=1+\dfrac{\tau_{1,1}}{a_{2,2}}\, r,
\end{equation}
the necessary and sufficient conditions, respectively, reads 
\begin{equation}\label{eq7}
\tau_{1,0}=a_{2,0},\qquad 
\left|\begin{array}{lll}
-\tau_{1,1}&a_{2,2}\\
-a_{2,0}& a_{2,1} -\tau_{1,1}
\end{array}\right|=0.
\end{equation}
For the second-order polynomial solution
 \begin{equation}\label{eq8}
y_2(r)=1+\dfrac{\tau_{1,1}}{a_{2,2}}\, r +\frac{2 a_{2,0}a_{2,2}-a_{2,1}\tau_{1,1}+\tau_{1,1}^2}{2a_{2,2}(a_{2,2}+ a_{3,2})}\,r^2,
\end{equation}
the necessary and sufficient conditions, respectively, reads 
\begin{equation}\label{eq9}
\tau_{1,0}=2a_{2,0},\qquad 
\left|\begin{array}{lll}
-\tau_{1,1}&a_{2,2}&0\\
-2a_{2,0}& a_{2,1} -\tau_{1,1}&2a_{3,2}+2a_{2,2}\\
0&-a_{2,0}&2a_{3,1}+2a_{2,1}-\tau_{1,1}
\end{array}\right|=0.
\end{equation}
For the third-order polynomial solution
 \begin{align}\label{eq10}
y_3(r)&=1+\dfrac{\tau_{1,1}}{a_{2,2}}\, r +\frac{3 a_{2,0}a_{2,2}-a_{2,1}\tau_{1,1}+\tau_{1,1}^2}{2a_{2,2}(a_{2,2}+ a_{3,2})}\,r^2\notag\\
&+
\frac{-6a_{2,0}a_{2,2} (a_{2,1}+a_{3,1})+a_{2,0} (7 a_{2,2}+4a_{3,2}) \tau_{1,1}+(a_{2,1}-\tau_{1,1}) (2 (a_{2,1}+a_{3,1})-\tau_{1,1})\tau_{1,1}}{6a_{2,2} (a_{2,2}+a_{3,2}) (a_{2,2}+2a_{3,2})}\,r^3,
\end{align}
the necessary and sufficient conditions, respectively, reads 
\begin{equation}\label{eq11}
\tau_{1,0}=3a_{2,0},\qquad 
\left|\begin{array}{llll}
-\tau_{1,1}&a_{2,2}&0&0\\
-3a_{2,0}& a_{2,1} -\tau_{1,1}&2a_{3,2}+2a_{2,2}&0\\
0&-2a_{2,0}&2a_{3,1}+2a_{2,1}-\tau_{1,1}&6a_{3,2}+3a_{2,2}\\
0&0&-a_{2,0}&6a_{3,1}+3a_{2,1}-\tau_{1,1}
\end{array}\right|=0.
\end{equation}
In general, for an $n$-degree polynomial solution of the confluent Heun equation, we have the following theorem.
\begin{theorem} \label{ThmII.1}
The necessary and sufficient conditions for an $n$-degree polynomial solution of equation \eqref{eq1}, respectively, are
\begin{equation}\label{eq11}
\tau_{1,0}=n\, a_{2,0},\qquad n=0,1,\dots,
\end{equation}
and the vanishes of the $(n+1)\times (n+1)$ determinant
\begin{align}\label{eq13}
\left|\begin{array}{lllllll}
-\tau_{1,1}&a_{2,2}&0&0&\dots& 0&0\\
-\tau_{1,0}& a_{2,1} -\tau_{1,1}&2a_{3,2}+2a_{2,2}&0&\dots&0&0\\
0&a_{2,0}-\tau_{1,0}&2a_{3,1}+2a_{2,1}-\tau_{1,1}&6a_{3,2}+3a_{2,2}&\dots&0&0\\
0&0&2a_{2,0}-\tau_{1,0}&6a_{3,1}+3a_{2,1}-\tau_{1,1}&\dots&0&0\\
\vdots&\vdots&\vdots&\ddots&\ddots&\vdots&\vdots\\
0&0&0&0&\ddots&(n-1)a_{2,0}-\tau_{1,0}&n(n-1)a_{3,1}+na_{2,1}-\tau_{1,1}\\
\end{array}\right|=0.
\end{align}
\end{theorem}
\vskip0.1true in
\noindent Theorem \ref{ThmII.1} generalized the two conditions on the standard confluent Heun equation to allow for the $n$-degree polynomial solutions \cite{fiz2010, fiz2011}.
  \vskip0.1true in
\noindent \textbf{Remark II.1:} If $a_{3,2}=a_{2,2}=0$, the differential equation \eqref{eq1} reduces to
\begin{align}\label{eq14}
a_{3,1}r^2\,\frac{d^2y}{dr^2}+r\,(a_{2,0}r+a_{2,1})\frac{dy}{dr}-(\tau_{1,0}r+\tau_{1,1})y=0,
\end{align}
that have the $n$-degree polynomial solutions expressed in terms of the confluent hypergeometric function as
\begin{align}\label{eq15}
y_{n;k}(r)=r^k {}_1F_1\left(k-n;2k+\frac{a_{2,1}}{a_{3,1}};-\frac{a_{2,0}}{a_{3,1}}\,r\right),\qquad k=0,1,\dots n,
\end{align}
 provided that:
$\tau_{1,0}=n\,a_{2,0}$ and $\tau_{1,1}=k(k-1)\,a_{3,1}+k\,a_{2,1}, \quad k=0,1,2,\dots,n,\quad n=0,1,2,\dots.
$

\section{Finite sequences of orthogonal polynomials}
\noindent Subject to the necessary and sufficient conditions \eqref{eq13}, the polynomial solutions \eqref{eq6}, \eqref{eq8} and \eqref{eq10} of the differential equation \eqref{eq1} can be generalized as
\begin{align}\label{eq16}
y_n(r)&=\sum\limits_{k=0}^n\dfrac{\mathpzc{P}_{k}^n(\tau_{1,1})}{k!\, a_{3,2}^k\,\left(\frac{a_{2,2}}{a_{3,2}}\right)_k}\,r^k,\quad n=0,1,2,\dots,
\end{align}
where $(\alpha)_n$ refer to the pochhammer symbol $(\alpha)_n=\alpha(\alpha+1)\dots(\alpha-n+1)$ satisfies $(-n)_k=0$ for any positive integers $k\geq n$.  
\vskip0.1true in
\noindent The coefficients $\mathpzc{P}_{k}^n(\zeta),~k=0,\dots,n$  are a sequence of $k$-degree polynomials in the parameter $\zeta=\tau_{1,1}$. This sequence of polynomials will be our focus in the present work. Indeed, for each $n$, the solution $y_n(r)$ of the differential equation \eqref{eq1} generate a finite sequence of polynomials  $\{\mathpzc{P}_{k}^j(\tau_{1,1})\}_{k=0}^n$ satisfying a number of very interesting properties. 
\vskip0.1true in
\subsection{Three-term recurrence relation}
\noindent The first property is that they satisfy a three-term recurrence relation 
\begin{equation}\label{eq17}
\mathpzc{P}_{k+1}^n(\zeta)+\left(k(k-1)a_{3,1}+ka_{2,1}-\zeta\right)\mathpzc{P}_k^n(\zeta)+k\,((k-1)a_{3,2}+a_{2,2})\left((k-1)a_{2,0}-\tau_{1,0}\right)\mathpzc{P}_{k-1}^n(\zeta)=0,\quad 0\leq k\leq n-1,
\end{equation}
initiated with $\mathpzc{P}_{-1}^n(\zeta)=0, ~ \mathpzc{P}_{0}^n(\zeta)=1.$ This recurrence relation can be obtained by the fact that $y_n(r)$ satisfies the differential equation \eqref{eq1}.
Using the necessary condition $\tau_{1,0}=n\, a_{2,0}$, equation \eqref{eq17} reduce to
\begin{equation}\label{eq18}
\mathpzc{P}_{k+1}^n(\zeta)=\left(\zeta-k(k-1)a_{3,1}-ka_{2,1}\right)\mathpzc{P}_k^n(\zeta)-k\,(k-n-1) a_{2,0}((k-1)a_{3,2}+a_{2,2})\mathpzc{P}_{k-1}^n(\zeta),~~\mathpzc{P}_{-1}^n(\zeta)=0, \quad \mathpzc{P}_{0}^n(\zeta)=1,
\end{equation}
for $k=0,1,2,\dots n-1$. Here, $n$ is a (fixed) nonnegative integer refer to the degree of the polynomial solution \eqref{eq16}. 

\vskip0.1true in
\noindent  From the general theory of orthogonal polynomials and according to Favard's Theorem  \cite{Favard,chihara,ismail}),  
$\{\mathpzc{P}_{k}^n(\zeta)\}_{k=0}^n$  form a sequence of orthogonal polynomials. The first few expressions of these polynomials reads
\begin{align*}
\mathpzc{P}_1^n(\zeta)&=\zeta,\\
\mathpzc{P}_2^n(\zeta)&=\zeta^2-a_{2,1}\zeta+n\,a_{2,0}\,a_{2,2},\\
\mathpzc{P}_3^n(\zeta)&=\zeta^3 -(3a_{2,1}+2a_{3,1})\zeta^2+\left(2 a_{2,1}^2+2 a_{2,1} a_{3,1}+a_{2,0} a_{2,2}n-2 (a_{2,2}+a_{3,2}) (a_{2,0}-a_{2,0}n)\right)\zeta-2na_{2,0}a_{2,2}(a_{2,1}+a_{3,1})\\
\mathpzc{P}_4^n(\zeta)&=\zeta^4-2(3a_{2,1}+4a_{3,1})\zeta^3+(11a_{2,1}^2+26 a_{2,1}a_{3,1}+12 a_{3,1}^2+2 a_{2,0} (a_{2,2} (3n-4)+a_{3,2} (4n-7)))\zeta^2\\
&-2 \left(3 (a_{2,1}+a_{3,1}) \left(a_{2,1}^2-2a_{2,0}a_{2,2}+2a_{2,1}a_{3,1}\right)- a_{2,0} a_{3,2} (9a_{2,1}+6 a_{3,1})+na_{2,0} (a_{2,2} (7a_{2,1}+10 a_{3,1})+6 (a_{2,1}+a_{3,1}) a_{3,2})\right)\zeta\\
&+3na_{2,0}a_{2,2} (2 (a_{2,1}+a_{3,1}) (a_{2,1}+2 a_{3,1})+a_{2,0}  (n-2)(a_{2,2}+2 a_{3,2})).
\end{align*}

\noindent An important consequence of the three-term recurrence formula \eqref{eq18} is the following Christoffel-Darboux formula.
\vskip0.1true in
\begin{theorem} For $\zeta\neq \zeta'$,
\begin{align}\label{eq19}
\sum_{j=0}^k\frac{\mathpzc{P}_{j}^n(\zeta)\mathpzc{P}_j^n(\zeta')}{j!\,
(a_{2,0} a_{3,2})^j\left(\frac{a_{2,2}}{a_{3,2}}\right)_j(-n)_j
}
&=\frac{\mathpzc{P}_{k+1}^n(\zeta)\mathpzc{P}_{k}^n(\zeta')-\mathpzc{P}_{k}^n(\zeta)
\mathpzc{P}_{k+1}^n(\zeta')}{k!\,
(a_{2,0} a_{3,2})^k\left(\frac{a_{2,2}}{a_{3,2}}\right)_k(-n)_k (\zeta-\zeta')}, \qquad k=1,\dots, n-1,
\end{align}
while for $\zeta'\rightarrow \zeta$
\begin{align}\label{eq20}
\sum_{j=0}^k\frac{\big(\mathpzc{P}_{j}^n(\zeta)\big)^2}{j!\,
(a_{2,0} a_{3,2})^j\left(\frac{a_{2,2}}{a_{3,2}}\right)_j(-n)_j
}
&=\frac{[\mathpzc{P}_{k+1}^n(\zeta)]'\mathpzc{P}_{k}^n(\zeta)-[\mathpzc{P}_{k}^n(\zeta)]'
\mathpzc{P}_{k+1}^n(\zeta)
}{k!\,
(a_{2,0} a_{3,2})^k\left(\frac{a_{2,2}}{a_{3,2}}\right)_k(-n)_k},\qquad k=1,\dots,n-1.
\end{align}
\end{theorem}
\begin{proof} The recurrence relations \eqref{eq18} for $\zeta$ and $\zeta'$ reads, respectively,
\begin{align*}
\mathpzc{P}_{k+1}^n(\zeta)&=\zeta\mathpzc{P}_k^n(\zeta)-[k(k-1)a_{3,1}+ka_{2,1}]\mathpzc{P}_k^n(\zeta)-k\,(k-n-1) a_{2,0}((k-1)a_{3,2}+a_{2,2})\mathpzc{P}_{k-1}^n(\zeta),\\
\mathpzc{P}_{k+1}^n(\zeta')&=\zeta'\mathpzc{P}_k^n(\zeta')-[k(k-1)a_{3,1}+ka_{2,1}]\mathpzc{P}_k^n
(\zeta')-k\,(k-n-1) a_{2,0}((k-1)a_{3,2}+a_{2,2})\mathpzc{P}_{k-1}^n(\zeta').
\end{align*}
Multiplying the first by $\mathpzc{P}_{k}^n(\zeta')$ and the second by $\mathpzc{P}_{k}^n(\zeta)$ then subtruct, the resulting equation reads
\begin{align*}
(\zeta-\zeta')\mathpzc{P}_{k}^n(\zeta')
\mathpzc{P}_k^n(\zeta)=Q_{k+1}(\zeta,\zeta')-\lambda_{k+1}Q_{k}(\zeta,\zeta')
\end{align*}
where $\lambda_{k+1}=k\,(k-n-1) a_{2,0}((k-1)a_{3,2}+a_{2,2}),$ and  $Q_{k+1}(\zeta,\zeta')=\mathpzc{P}_{k+1}^n(\zeta)\mathpzc{P}_{k}^n(\zeta')-\mathpzc{P}_{k}^n(\zeta)
\mathpzc{P}_{k+1}^n(\zeta')$. Thus, recursively over $k$ we have
\begin{align*}
(\zeta-\zeta')\mathpzc{P}_{k}^n(\zeta)\mathpzc{P}_k^n(\zeta')&=Q_{k+1}(\zeta,\zeta')-\lambda_{k+1}
Q_{k}(\zeta,\zeta')\\
(\zeta-\zeta')\mathpzc{P}_{k-1}^n(\zeta)\mathpzc{P}_{k-1}^n(\zeta')&=Q_{k}(\zeta,\zeta')-\lambda_{k}
Q_{k-1}(\zeta,\zeta')\\
\dots&=\dots\\
(\zeta-\zeta')\mathpzc{P}_0^n(\zeta)\mathpzc{P}_0^n(\zeta')&=Q_1(\zeta,\zeta').
\end{align*}
From which, it is straightforward to obtain
\begin{align*}
(\zeta-\zeta')&\bigg[\mathpzc{P}_{k}^n(\zeta)\mathpzc{P}_k^n(\zeta')+
\lambda_{k+1}\mathpzc{P}_{k-1}^n(\zeta)\mathpzc{P}_{k-1}^n(\zeta')+
\lambda_{k+1}\lambda_{k}\mathpzc{P}_{k-2}^n(\zeta)\mathpzc{P}_{k-2}^n(\zeta')+
\lambda_{k+1}\lambda_{k}\lambda_{k-1}\mathpzc{P}_{k-3}^n(\zeta)\mathpzc{P}_{k-3}^n(\zeta')
+\dots\notag\\
&+
\lambda_{k+1}\lambda_{k}\lambda_{k-1}\lambda_{k-2}\dots \lambda_2\mathpzc{P}_0^n(\zeta)\mathpzc{P}_0^n(\zeta')\bigg]=Q_{k+1}(\zeta,\zeta')
\end{align*}
Divide  both sides by $(\zeta-\zeta')\lambda_{k+1}\lambda_{k}\lambda_{k-1}\lambda_{k-2}\dots \lambda_2$ and sum over $k$ we finally gets
\begin{align*}
\sum_{j=0}^k\frac{\mathpzc{P}_{j}^n(\zeta)\mathpzc{P}_j^n(\zeta')}{\lambda_{j+1}\lambda_{j}\lambda_{j-1}\dots \lambda_2}
&=(\lambda_{k+1}\lambda_{k}\lambda_{k-1}\dots\lambda_2 )^{-1}\frac{\mathpzc{P}_{k+1}^n(\zeta)\mathpzc{P}_{k}^n(\zeta')-\mathpzc{P}_{k}^n(\zeta)
\mathpzc{P}_{k+1}^n(\zeta')}{\zeta-\zeta'}.
\end{align*}
Formula \eqref{eq17} then follows using $$\lambda_{k+1}\lambda_{k}\lambda_{k-1}\dots\lambda_2=\prod_{j=2}^{k+1}\lambda_j=k!\,
(a_{2,0} a_{3,2})^k\left(\frac{a_{2,2}}{a_{3,2}}\right)_k(-n)_k,\qquad k=1,\dots, n,$$
and finally the formula \eqref{eq18} follows by evaluating the limit of both sides of \eqref{eq17} as $\zeta'\rightarrow \zeta$.
\end{proof}
\begin{theorem}
The zeros $\{\zeta_j\}_{j=0}^k$ of the orthogonal polynomial $\mathpzc{P}_{k+1}^n(\zeta), k=0,\dots,n-1$ are 
eigenvalues of the following tridiagonal matrix
\begin{equation}\label{eq21}
\begin{pmatrix}
\alpha_0&1& & & &\\
\beta_1&\alpha_1&1& & & \\
 &\beta_2&\alpha_2&1& & \\
& & \ddots&\ddots &\ddots & \\
& & &   \beta_{k-1}&\alpha_{k-1}& 1\\
& & &  &\beta_k& \alpha_k
\end{pmatrix}
\end{equation}
where
$\alpha_j=j(j-1)a_{3,1}+ja_{2,1},$ for $ j\geq 0$ and $\beta_j=j\,(j-n-1) a_{2,0}((j-1)a_{3,2}+a_{2,2}),$ for $ j\geq 1.$
\end{theorem} 
\begin{proof}
We first write the recurrence relation \eqref{eq18} as
\begin{equation}\label{eq22}
\zeta\mathpzc{P}_{j}^n(\zeta)=\beta_j\mathpzc{P}_{j-1}^n(\zeta) +\alpha_j\mathpzc{P}_j^n(\zeta)+\mathpzc{P}_{j+1}^n(\zeta),\qquad j\geq 0,
\end{equation}
We now take $j=0,1,2,\dots k$ to form a system with a matrix form
\begin{equation}\label{eq23}
\begin{pmatrix}
\zeta \mathpzc{P}_{0}^n(\zeta)\\
\zeta \mathpzc{P}_{1}^n(\zeta)\\
\zeta \mathpzc{P}_{2}^n(\zeta)\\
\vdots\\
\zeta \mathpzc{P}_{k-1}^n(\zeta)\\
\zeta \mathpzc{P}_{k}^n(\zeta)
\end{pmatrix}=
\begin{pmatrix}
\alpha_0&1& & & &\\
\beta_1&\alpha_1&1& & & \\
 &\beta_2&\alpha_2&1& & \\
& & \ddots&\ddots &\ddots & \\
& & &   \beta_{k-1}&\alpha_{k-1}& 1\\
& & &  &\beta_k& \alpha_k
\end{pmatrix}\begin{pmatrix}
\mathpzc{P}_{0}^n(\zeta)\\
\mathpzc{P}_{1}^n(\zeta)\\
\mathpzc{P}_{2}^n(\zeta)\\
\vdots\\
\mathpzc{P}_{k-1}^n(\zeta)\\
\mathpzc{P}_{k}^n(\zeta)
\end{pmatrix}
\end{equation}
The system at the zeros $\zeta=\zeta_j$ becomes
\begin{equation}\label{eq24}
\begin{pmatrix}
\alpha_0-\zeta_j&1& & & &\\
\beta_1&\alpha_1-\zeta_j&1& & & \\
 &\beta_2&\alpha_2-\zeta_j&1& & \\
& & \ddots&\ddots & & \\
& & &   \beta_{k-1}&\alpha_{k-1}-\zeta_j& 1\\
& & &  &\beta_k& \alpha_k-\zeta_j
\end{pmatrix}\begin{pmatrix}
\mathpzc{P}_{0}^n(\zeta_j)\\
\mathpzc{P}_{1}^n(\zeta_j)\\
\mathpzc{P}_{2}^n(\zeta_j)\\
\vdots\\
\mathpzc{P}_{k-1}^n(\zeta_j)\\
\mathpzc{P}_{k}^n(\zeta_j),
\end{pmatrix}=0,\quad 0\leq j\leq k.
\end{equation}
Hence, $\{\zeta_j\}_{j=0}^k$ are eigenvalues of the tridiagonal matrix \eqref{eq19}.
\end{proof}
\subsection{Factorization Property}
\noindent Another interesting property of the polynomials 
$\{\mathpzc{P}_k^n(\zeta)\}_{k=0}^n$, next to being an orthogonal sequence, is that when the
parameter $n$
takes positive integer values the polynomials exhibit a
factorization property similar to the Bender-Dunne orthogonal polynomials \cite{bender1996, saad2006,agn}. 
\vskip0.1true in
\noindent Clearly, The factorization occurs because the third-term
in the recursion relation \eqref{eq18}
vanishes when $k=n+1$, so that all subsequent polynomials have a common factor  $\mathpzc{P}_{n+1}^n(\zeta)$ called a \emph{critical polynomial}. To illustrate this factorization property, we examine it in the case of $n=1$:
\begin{align*}
\mathpzc{P}_1^1(\zeta)&=\zeta,\\
\mathpzc{P}_2^1(\zeta)&=\zeta^2-a_{2,1}\zeta +a_{2,0} a_{2,2},\\
\mathpzc{P}_3^1(\zeta)&=(\zeta-2 a_{2,1}-2a_{3,1}) \mathpzc{P}_2^1(\zeta),\\
\mathpzc{P}_4^1(\zeta)&=\left(6 a_{2,1}^2-3a_{2,0} a_{2,2}+18a_{2,1} a_{3,1}+12 a_{3,1}^2-6 a_{2,0}a_{3,2}-5 a_{2,1}\zeta-8 a_{3,1}\zeta+\zeta^2\right)\mathpzc{P}_2^1(\zeta).
\end{align*}
Here, the critical polynomial $P_2^1(\zeta)$ factorize every other polynomial $P_{k+2}^n(\zeta), k=1,2,\dots$. 
As another illustration, we note for $n=2$: 
\begin{align*}
\mathpzc{P}_1^2(\zeta)&=\zeta,\\
\mathpzc{P}_2^2(\zeta)&=\zeta^2-a_{2,1}\zeta +2a_{2,0} a_{2,2},\\
\mathpzc{P}_3^2(\zeta)&=-4 a_{2,0} a_{2,1}a_{2,2}-4a_{2,0}a_{2,2}a_{3,1}+\left(2a_{2,1}^2+4a_{2,0} a_{2,2}+2 a_{2,1}a_{3,1}+2a_{2,0}a_{3,2}\right)\zeta-(3 a_{2,1}+2a_{3,1}) \zeta^2+\zeta^3,\\
\mathpzc{P}_4^2(\zeta)&=(\zeta-3a_{2,1}-6a_{3,1})\mathpzc{P}_3^2(\zeta),\\
\mathpzc{P}_5^2(\zeta)&=(12 a_{2,1}^2 - 4 a_{2,0} a_{2,2} + 60 a_{2,1} a_{3,1} + 72 a_{3,1}^2 - 
 12 a_{2,0} a_{3,2} - (7 a_{2,1}+18 a_{3,1})\zeta + \zeta^2) \mathpzc{P}_3^2(\zeta).
\end{align*}
with $P_3^n(\zeta)$ being a critical polynomial factorize every other polynomial $P_{k+3}^n(\zeta), k=1,2,\dots$.  Indeed, all the polynomials
$\mathpzc{P}_{k+n+1}^n(\zeta)$, beyond some critical polynomial $\mathpzc{P}_{n+1}^n(\zeta)$, factored into the product 
\begin{align}\label{eq25}
\mathpzc{P}_{k+n+1}^n(\zeta)=\mathcal Q_k^n(\zeta)\mathpzc{P}_{n+1}^n(\zeta),\qquad k=0,1,\dots.
\end{align}
Interestingly, the quotient polynomials $\{\mathcal Q_k^n(\zeta)\}_{k\geq 0}$  form an infinite sequence of orthogonal polynomials. To prove this claim, we substitute \eqref{eq25} into \eqref{eq18} where we re-index the polynomials to eliminate the common factor $\mathpzc{P}_{n+1}^n(\zeta)$ from both sides. The recurrence relation \eqref{eq18} then reduces to a three-term recurrence relation for the polynomials $\{\mathcal Q_k^n(\zeta)\}_{k\geq 0}$ that reads
\begin{equation}\label{eq26}
\mathcal Q_k^n(\zeta)=\left(\zeta- (k+n)((k+n-1)a_{3,1}+a_{2,1})\right)\mathcal Q_{k-1}^n(\zeta)-\,(k-1)(k+n) a_{2,0}((k+n-1)a_{3,2}+a_{2,2})\mathcal Q_{k-2}^{n}(\zeta),\quad k\geq 1,
\end{equation}
initiated with $\mathcal Q_{-1}^n(\zeta)=0,$ and $\mathcal Q_{0}^n(\zeta)=1$. Hence, the quotient polynomials $\mathcal Q_k^n(\zeta)$
also form a new sequence of orthogonal polynomials for each value of $n$. The first few polynomials of $\mathcal Q_k^n(\zeta)$ can to explicitly evaluate to reads
\begin{align*}
\mathcal Q_1^n(\zeta)&=\zeta-(n+1)(a_{2,1}+na_{3,1}),\\
\mathcal Q_2^n(\zeta)&=\zeta^2-(2a_{3,1}(1+n)^2+a_{2,1}(2 n+3))\zeta+(n+2)\left[(1+n) (a_{2,1}+a_{3,1} n) (a_{2,1}+a_{3,1}+a_{3,1} n)-a_{2,0} (a_{2,2}+a_{3,2}+a_{3,2} n)\right].
\end{align*}
\vskip0.1true in
\noindent The Christoffel-Darboux formula for this sequence of orthogonal polynomials reads
\begin{align}\label{eq27}
\sum_{j=0}^k\frac{{\mathcal Q}_j^n(\zeta){\mathcal Q}_j^n(\zeta')}{j!(n+2)_j(a_{2,0}a_{3,2})^j\left(n+1+\frac{a_{2,2}}{a_{3,2}}\right)_j}=\frac{{\mathcal Q}_{k+1}^n(\zeta){\mathcal Q}_{k}^n(\zeta')-{\mathcal Q}_{k}^{n}(\zeta){\mathcal Q}_{k+1}^n(\zeta')}{k!(n+2)_k(a_{2,0}a_{3,2})^k\left(n+1+\frac{a_{2,2}}{a_{3,2}}\right)_k(\zeta-\zeta')},
\end{align}
and as $\zeta'\rightarrow \zeta'$,
\begin{align}\label{eq28}
\sum_{j=0}^k\frac{\left({\mathcal Q}_j^n(\zeta)\right)^2}{j!(n+2)_j(a_{2,0}a_{3,2})^j\left(n+1+\frac{a_{2,2}}{a_{3,2}}\right)_j}=\frac{[{\mathcal Q}_{k+1}^n(\zeta)]'{\mathcal Q}_{k}^n(\zeta)-[{\mathcal Q}_{k}^{n}(\zeta)]'{\mathcal Q}_{k+1}^n(\zeta)}{k!(n+2)_k(a_{2,0}a_{3,2})^k\left(n+1+\frac{a_{2,2}}{a_{3,2}}\right)_k}.
\end{align}

\subsection{Moments and weight functions}
\noindent From the general theory of orthogonal polynomials,  there exists a
certain weight function $\mathpzc{W}(\xi)$ normalized as 
\begin{equation}\label{eq29}
\int_{\mathpzc{S}} \mathpzc{W}(\xi)d\xi=1,
\end{equation}
for which
\begin{equation}\label{eq30}
\int_{\mathcal S}  \mathpzc{P}_{k}^n(\zeta)\,\mathpzc{P}_{k'}^n(\zeta)\,\mathpzc{W}(\zeta)\,d\zeta=p_kp_{k'}\delta_{kk'},\qquad k,k'=0,1,\dots,n,
\end{equation}
where $\mathpzc{S}$ is the support of the measure $\mathpzc{W}(\zeta)\,d\zeta$ on the real line and $\delta_{kk'}$ is Kronecker's symbol.  
\begin{theorem} \label{Thm3.2} For the monic orthogonal polynomials $\{ \mathpzc{P}_{k}^n(\zeta) \}_{k=0}^n$
\begin{equation}\label{eq31}
\int_{\mathcal S}  \zeta^{j} \mathpzc{P}_{k}^n(\zeta)\,\mathpzc{W}(\zeta)\,d\zeta= p_k^2\,\delta_{kj},\qquad\forall j\leq k.
\end{equation}
\end{theorem}
\begin{proof}
By writing $p_{k'}^n(\zeta)=\sum_{j=0}^{k'} \alpha_{k';j}\zeta^j$, where $\alpha_{k';k'}=1$, equation \eqref{eq22} yields
\begin{equation}\label{eq32}
\sum_{j=0}^{k'}\alpha_{k';j} \int_{\mathcal S} \zeta^j \mathpzc{P}_{k}^n(\zeta)\,\mathpzc{W}(\zeta)\,d\zeta=p_kp_{k'}\delta_{kk'},\qquad k,k'=0,1,\dots,n,\qquad a_{k';k'}=1.
\end{equation}
For $k'=0,1,\dots,n$, equation \eqref{eq32} generate a linear system that reduce to
\begin{align}\label{eq33}
\begin{pmatrix}
\int_{\mathcal S}  \mathpzc{P}_{k}^n(\zeta)\mathpzc{W}(\zeta)d\zeta\\ \\
\int_{\mathcal S}  \zeta\mathpzc{P}_{k}^n(\zeta)\mathpzc{W}(\zeta)d\zeta\\ \\
\int_{\mathcal S}  \zeta^2\mathpzc{P}_{k}^n(\zeta)\mathpzc{W}(\zeta)d\zeta\\ \\
\int_{\mathcal S}  \zeta^3\mathpzc{P}_{k}^n(\zeta)\mathpzc{W}(\zeta)d\zeta\\
\vdots\\
\int_{\mathpzc{S}}  \zeta^{k'}\mathpzc{P}_{k}^n(\zeta)\mathpzc{W}(\zeta)d\zeta
\end{pmatrix}=p^k\begin{pmatrix}
1&0&0&0&\dots& 0\\
\alpha_{1;0}^n& 1&0&0&\dots&0\\ \\
\alpha_{2;0}^n&\alpha_{2;1}^n&1&0&\dots&0\\ \\
\alpha_{3;0}^n&\alpha_{3;1}^n&\alpha_{3;2}^n&1&\dots&0\\ \\
\alpha_{4;0}^n&\alpha_{4;1}^n&\alpha_{4,2}^n&\alpha_{4,3}^n&\dots&0\\ \\
\vdots&\vdots&\vdots&\vdots&\dots&\vdots\\ \\
\alpha_{k';0}^n&\alpha_{k';1}^n&\alpha_{k';2}^n&\alpha_{k';3}^n&\dots&1
\end{pmatrix}_{(k'+1)\times (k'+1)}^{-1}\times\begin{pmatrix}
p_0\delta_{k0}\\ \\
p_1\delta_{k1}\\ \\
p_2\delta_{k2}\\ \\
p_3\delta_{k3}\\ \\
\vdots\\
p_{k'}\delta_{kk'}
\end{pmatrix}_{(k'+1)\times 1}.
\end{align}
The inverse of the lower triangular matrix is again lower triangular matrix with ones on the main diagonal. This argument concludes our assertion \eqref{eq27}. . 
\end{proof}
\vskip0.1true in
\noindent From Theorem \ref{Thm3.2}, we also conclude that
\begin{equation}\label{eq34}
\int_{\mathcal S}  \mathpzc{W}(\zeta)\,d\zeta= 1, \quad \int_{\mathcal S}  \mathpzc{P}_{k}^n(\zeta)\,\mathpzc{W}(\zeta)\,d\zeta= 0\quad (k\geq 1), \qquad \int_{\mathcal S}  \mathpzc{P}_{j}^n(\zeta)\mathpzc{P}_{k}^n(\zeta)\,\mathpzc{W}(\zeta)\,d\zeta= 0,\quad (0\leq j<k\leq n).
\end{equation}
\begin{theorem}\label{Thm3.4}
The norms of all polynomials
$\mathpzc{P}_{k}^n(\xi)$ with $k\geq n+1$ vanish.
\end{theorem}
\begin{proof}
\noindent The norms $p_k$ can be evaluated using the recurrence relations \eqref{eq18} after multiplying throughout by $\zeta^{k-1}\mathpzc{W}(\zeta)$ and taking integral over $\zeta$, we obtain, because of  Theorem  \ref{Thm3.2}, a simple two-term
recursion relation for the squared norm
\begin{equation*}
\mathcal G_k=k\,(k-1-n)a_{2,0}((k-1)a_{3,2}+a_{2,2})\mathcal G_{k-1},
\end{equation*}
with a solution given by
\begin{equation}\label{eq35}
\mathcal G_k=|q_k|^2=k!\, (a_{2,0}a_{3,2})^k (-n)_k\, \left(\frac{a_{2,2}}{a_{3,2}}\right)_k,
\end{equation}
which is equal to zero for $k=n+1,n+2,\dots$.
\end{proof}
\noindent As a result of Theorem \ref{Thm3.4}, we note
\begin{align}
\int_{\mathcal S} \left[\mathpzc{P}_{k}^n(\zeta)\right]^2\,\mathpzc{W}(\zeta)\,d\zeta&=k!\, (a_{2,0}a_{3,2})^k (-n)_k\, \left(\frac{a_{2,2}}{a_{3,2}}\right)_k,\quad 0\leq k\leq n,\label{eq36}\\
\int_{\mathcal S}  \zeta\, \mathpzc{P}_{k}^n(\zeta)\mathpzc{P}_{k-1}^n(\zeta)\,\mathpzc{W}(\zeta)\,d\zeta&=k!\, (a_{2,0}a_{3,2})^k (-n)_k\, \left(\frac{a_{2,2}}{a_{3,2}}\right)_k,\qquad k=1,2,\dots,\label{eq37}\\
\int_{\mathcal S}  \zeta\, [\mathpzc{P}_{k}^n(\zeta)]^2\,\mathpzc{W}(\zeta)\,d\zeta&=k((k-1)a_{3,1}+a_{2,1})k!\, (a_{2,0}a_{3,2})^k (-n)_k\, \left(\frac{a_{2,2}}{a_{3,2}}\right)_k,\qquad k=0,1,2,\dots.\label{eq38}
\end{align}
\vskip0.1true in
\noindent Although we don't have a general formula for evaluating the moments 
\begin{align}\label{eq39}
\mu_k =\int_{\mathcal S}\zeta^k W(\zeta)d\zeta,\quad k=0,1,\dots,n,
\end{align}
it is possible to compute $\mu_k$ recursively using Theorem \ref{Thm3.2} and the recurrence relation \eqref{eq16}. It is not difficult to show that 
\begin{align}\label{eq40}
\int_{\mathcal S}W(\zeta)d\zeta=1,~\int_{\mathcal S}\zeta \,W(\zeta)d\zeta=0,~\int_{\mathcal S}\zeta^2\,W(\zeta)d\zeta= p_1^2=-n\,a_{2,0}a_{2,2},\quad a_{2,0}a_{2,2}<0, \end{align}
and for all $k\geq 2$,
\begin{equation}\label{eq41}
\int_{\mathcal S}\zeta \mathpzc{P}_k^n(\zeta) W(\zeta)d\zeta=0.
\end{equation}
For example,
\begin{align}\label{eq42}
\int_{\mathcal S}\zeta^3\,  W(\zeta)\,d\zeta=a_{2,1}p_1^2,\qquad 
\int_{\mathcal S}\zeta^4\,  W(\zeta)\, d\zeta=p_2^2+(a_{2,1}-na_{2,0}a_{2,2})p_1^2,\quad \mbox{and so forth.}
\end{align}\vskip0.1true in
\begin{theorem}
The norms of all polynomials 
$\mathpzc{Q}_{k}^n(\xi)$ reads
\begin{equation}\label{eq43}
\mathcal G_k^{\mathcal Q}=k!\,(n+2)_k\,(a_{2,0}a_{3,2})^k\left(n+1+\frac{a_{2,2}}{a_{3,2}}\right)_k.
\end{equation}
\end{theorem}
\begin{proof} The proof follows by multiplying the recurrence relation  \eqref{eq26} by $\zeta^{k-2}W(\zeta)$ and integrate over $\zeta$. This procedure implies the two-term recurrence relation 
 \begin{equation}\label{eq44}
\mathcal G_{k}^{\mathcal Q}=k\,(k+n+1) a_{2,0}((k+n)a_{3,2}+a_{2,2})\mathcal G_{k-1}^{\mathcal Q},\qquad k\geq 1
\end{equation}
where $\mathcal G_{k}^{\mathcal Q}=\int_{\mathcal S} |\mathcal Q_{k}^n(\zeta)|^2 W(z)dz=\int_{\mathcal S}\zeta^{k}Q_{k}^n(\zeta)d\zeta$.
with a solution given \eqref{eq43}.
\end{proof}
\section{Example: Hyperbolic potentials}
\noindent Recently, Downing \cite{downing2013} studied a class of potentials characterized  by two physical parameters $V_0$ and $d$ shaping the potential 
depth and width respectively
\begin{align}\label{eq45}
V(x)=-V_0\frac{\sinh^4(x/d)}{\cosh^6(x/d)},\quad -\infty<x<\infty.
\end{align}
In this section, we consider first a more general class of potentials given by \cite{zhang}
\begin{align}\label{eq46}
V(x)=-V_0\frac{\sinh^p(x/d)}{\cosh^q(x/d)},\quad -\infty<x<\infty, \quad V_0>0, \quad p\geq 0,\quad q>p.
\end{align}
The condition $p\geq 0$ required to avoid a singularity at $x=0$ and he condition $p<q$ is to ensure the existence of a bound-below potential support the presence of at least one bound state.  Indeed, it is well known  \cite{Chadan}: \emph{Given a non-positive potential $V\leq 0$ with 
\begin{align}\label{eq47}
\int_{-\infty}^\infty V(x)dx<0,
\end{align}
then there exists a bound state (independent of the potential depth and range) with energy $E<0$ 
for the Hamiltonian 
$H=-\frac{\hbar^2}{2m}\frac{d^2}{dx^2}+V(x).
$} For the potential \eqref{eq46}, it is not difficult to prove that
$$\int_{-\infty}^\infty V(x)dx=-\frac{\left(1+(-1)^p\right) V_0 \Gamma\left(\frac{1+p}{2}\right) \Gamma\left(\frac{1}{2} (q-p)\right)}{2\Gamma\left(\frac{1+q}{2}\right)}<0,\qquad p\geq 0,\quad p<q.
$$
The stationary one-dimensional Schr\"odinger equation for a nonrelativistic particle of mass $m$ and energy $E$ in a hyperbolic potential $V(x)$ written as follows 
\begin{align}\label{eq48}
-\frac{\hbar^2}{2m}\dfrac{d^2}{dx^2}\psi(x)-V_0\frac{\sinh^p(x/d)}{\cosh^q(x/d)}\psi(x)=E\,\psi(x),\quad V_0>0,\quad -\infty<x<\infty.
\end{align}
can be written using the substitution $z=x/d$ as 
\begin{align}\label{eq49}
-\dfrac{d^2}{dx^2}\psi(z)-d^2 U_0\frac{\sinh^p(z)}{\cosh^q(z)}\psi(z)=\varepsilon d^2\,\psi(z),\quad -\infty<z<\infty.
\end{align}
where
$U_0=2mV_0/\hbar^2$ and $\varepsilon=2mE/\hbar^2$. Upon making the change of variable 
$\eta={1}/{\cosh^2(z)}
$ such that the domain $-\infty<z<\infty$ maps to $0<\eta<1$, equation \eqref{eq49} yields
\begin{align}\label{eq50}
4\eta^2(1-\eta)\dfrac{d^2\psi(\eta)}{d\eta^2}+\left[4\eta-6\eta^2
\right]\frac{d\psi(\eta)}{d\eta}+\left(\varepsilon d^2+ d^2 U_0{\eta^{(q-p)/2}(1-\eta)^{p/2}}\right)\psi(\eta)=0,\qquad 0<\eta<1, \quad p\geq 0,\quad q> p.
\end{align}
It is straightforward to verify that $\eta=0$ is a regular singular point with exponents $\{\pm d\sqrt{-\varepsilon }/{2}\}$ where $ \varepsilon <0$
and the singular point 
$\eta=1$ is also a regular singular point with exponents
$s= \{0,1/2\}$ utilizing the symmetric and anti-symmetric solutions \cite{downing2013}, respectively. From the classical theory of ordinary differential equation, we know that if $\eta=\infty$ is a regular singular point as well, then the differential equation \eqref{eq50} is a hypergeometric-type equation with exact solutions expressed in terms of the Gauss hypergeometric function ${}_2F_1$. This claim can be be easily verified in the case of $p=0$ and $q=2$ where the differential equation \eqref{eq50} turn to be Fuchsian with the three regular singular points at $\eta=0,1,$ and $\infty$.  
\vskip0.1true in
\noindent The point $\eta=\infty$ is an irregular singular point only if $p>2$ and $q>p$.  A reason for $V(x)=-V_0\sinh^4(\eta)/\cosh^6(\eta)$ to be a quasi-exactly solvable potential that will be our focus  because of the recent interest. In this case, equation \eqref{eq50} reads
\begin{align}\label{eq51}
\eta^2\,(1-\eta)\,\psi''(\eta)+\eta\left[1-\frac32\eta
\right]\psi'(\eta)+\frac14\left(\varepsilon d^2+ U_0\,d^2 \eta\,(1-\eta)^2\right)\,\psi(\eta)=0.
\end{align}
Because of the asymptotic behavior at the singular points, the exact solutions takes the form
\begin{align}\label{eq52}
\psi(\eta)=\eta^{\beta/2} (1-\eta)^s e^{-\alpha\eta/2}f(\eta),\qquad \beta,~\alpha>0,~s=0,1/2,\qquad 0<\eta<1,
\end{align}
we obtain, for the function $f(\eta)$, the following differential equation
\begin{align}\label{eq53}
f''(\eta)&
+\left[\frac{1+4s}{2(\eta-1)}+\frac{1+\beta}{\eta}-\alpha\right]f'(\eta)+\left[\frac{-\alpha+\beta+\beta^2+4s-4\alpha s+4\beta s }{4 (\eta-1)}+\frac{-2\alpha-\beta-2\alpha\beta-\beta^2-4s-4\beta s+\alpha^2}{4 \eta}
\right]f(\eta)=0
\end{align}
where $s=0,1/2$, $\beta=d\sqrt{-\varepsilon}$ and $\alpha=d\sqrt{U_0}$. Equation \eqref{eq53} can be written as
\begin{align}\label{eq54}
4\eta(\eta-1)f''(\eta)
&+\left[-4\alpha\eta^2+(6+4\alpha+4\beta+8s)\eta-4(1+\beta)\right]f'(\eta)\notag\\
&+\left[(\alpha(\alpha-2\beta-4s-3)\eta-\alpha^2 + 2 \alpha (1 + \beta) + (1 + \beta) (\beta + 4 s)
\right]f(\eta)=0
\end{align}
Using Theorem \ref{ThmII.1}, the necessary condition for the differential equation \eqref{eq54} to admit the polynomial solutions is
$$\alpha-2\beta=4n+4s+3$$
that implies the following eigenvalue spectra:
\begin{align}\label{eq55}
\varepsilon_n= -\frac{\left(4 n+4s+3-d \sqrt{U_0}\right)^2}{4 d^2},\qquad d^2U_0>(4n+4s+3)^2.
\end{align}
Again using Theorem \ref{ThmII.1}, the sufficient condition follows by the determinant
\begin{align}\label{eq56}
\left|\begin{array}{lllllllll}
\mu_0&\delta_1&0&0&0&\dots&0& 0&0\\
\gamma_1&\mu_1&\delta_2&0&0&\dots&0&0&0\\
0&\gamma_2&\mu_2&\delta_3&0&\dots&0&0&0\\
0&0&\gamma_3&\mu_3&\delta_4&\dots&0&0&0\\
\vdots&\vdots&\vdots&\ddots&\ddots&\ddots&\vdots&\vdots&\vdots\\
0&0&0&0&0&\dots&\gamma_{n-1}&\mu_{n-1}&\delta_{n}\\
0&0&0&0&0&\dots&0&\gamma_n&\mu_n\\
\end{array}\right|=0,
\end{align}
where
\begin{align*}
\mu_j&=-\alpha^2+2\alpha(2j+\beta+1)+(2j+\beta+1)(2j+\beta+4s),\qquad 
\gamma_j=4\alpha(1-j+n),\qquad
\delta_j=-4j(\beta+j).
\end{align*}
The exact solutions are given explicitly by the equation
\begin{align}\label{eq56}
\psi_n(\eta)=\, \eta^{d\sqrt{-\varepsilon_n}}\, (1-\eta)^s\, e^{-d\sqrt{U_0}\eta/2}\sum\limits_{k=0}^n\dfrac{\mathpzc{P}_{k}^n(\zeta)}{k!\,\left(1+\beta\right)_k}\,\left(-\frac{\eta}{4}\right)^k,\quad n=0,1,2,\dots,
\end{align}
where the coefficients $\mathpzc{P}_{k}^n(\zeta)$, with $\zeta\equiv\tau_{1,1}=\left(1+d \sqrt{-\varepsilon_n}\right) \left(d \sqrt{-\varepsilon_n}+4 s\right)+2 d \left(1+d \sqrt{-\varepsilon}\right) \sqrt{U_0}-d^2 U_{0}$  are evaluated using the following three-term recurrence relation
\begin{equation}\label{eq57}
\mathpzc{P}_{k+1}^n(\zeta)=\left(\zeta-2 k (1 + 2 \alpha + 2 \beta + 2 k + 4 s)\right)\mathpzc{P}_k^n(\zeta)-16k(k-n-1)\alpha(\beta+k)\mathpzc{P}_{k-1}^n(\zeta),~~\mathpzc{P}_{-1}^n(\zeta)=0, ~ \mathpzc{P}_{0}^n(\zeta)=1,
\end{equation}
with the square norm given by
\begin{align}\label{eq58}
\int_{\mathcal S} \left[\mathpzc{P}_{k}^n(\zeta)\right]^2\,\mathpzc{W}(\zeta)\,d\zeta&=2^{4k}k!\, \alpha^k (-n)_k\, \left(1+\beta\right)_k,\quad 0\leq k\leq n.
\end{align}
The first few solutions are:
\begin{itemize}
\item For $n=0$, 
we have
\begin{equation}
\varepsilon_0= -\frac{\left(4s+3-d \sqrt{U_0}\right)^2}{4 d^2},
\qquad
f_0(\eta)=1,
\end{equation}
subject to the root finding of the equation
\begin{equation}
7 d^2 U_0-12 d \sqrt{U_0}-3(4s+1)(4s+5)=0.
\qquad\mbox{where}\qquad
U_0>(4s+3)^2/d^2.
\end{equation}
The critical polynomial reads
\begin{align*}
\mathpzc{P}_{1}^0(\alpha,\beta;s)
=-\alpha^2 + 2 \alpha (1 + \beta) + (1 + \beta) (\beta + 4 s).
\end{align*}
\item For $n=1$, we have
\begin{equation}
\varepsilon_1= -\frac{\left(4s+7-d \sqrt{U_0}\right)^2}{4 d^2},
\qquad
f_1(\eta)=1-\left(\frac{\alpha^2 - 2 \alpha (1 + \beta)- (1 + \beta) (\beta + 4 s)}{4(1+\beta)}\right)\eta.
\end{equation}
subject to the root finding of the equation
\begin{align}
\alpha^4-4(2+\beta)\alpha^3&+2  (3 + 5 \beta + \beta^2 - 4 (2 + \beta) s)\alpha^2+4(1 + \beta) (3 + 4 \beta + \beta^2 + 4 n + 4 (3 + \beta) s) \alpha \notag\\
&+(1 + \beta) (3 + \beta) (\beta + 4 s) (2 + \beta + 4 s)=0
\end{align}
where $\alpha=d\sqrt{U_0},~\beta=(7+4s-d\sqrt{U_0})/2$ and $d^2U_0>(7+4s)^2$. The critical polynomial in this case reads
\begin{align*}
\mathpzc{P}_{2}^1(\alpha,\beta;s)=16 \alpha (1 + \beta) + (-\alpha^2 + 
    2 \alpha (1 + \beta) + (1 + \beta) (\beta + 4 s)) (-\alpha^2 + 
    2 \alpha (1 + \beta) + (1 + \beta) (\beta + 4 s) - 
    2 (3 + 2 \alpha + 2 \beta + 4 s)).
\end{align*}

\item For $n=2$, we have
\begin{align}
\varepsilon_2&= -\frac{\left(4s+11-d \sqrt{U_0}\right)^2}{4 d^2},\notag\\
f_2(\eta)&=1-\frac{\alpha^2 - 2 \alpha (1 + \beta)- (1 + \beta) (\beta + 4 s)}{4(1+\beta)}\eta\notag\\
&+\frac{32 \alpha (1 +\beta) - (6+ \alpha^2 + 2 \alpha (1- \beta)+ 4 s +
    \beta (3 - \beta - 4 s)) (-\alpha^2 + 
    2 \alpha (1 + \beta) + (1 + \beta) (\beta + 4 s))}{32(1+\beta)(2+\beta)}\eta^2,
\end{align}
subject to the the root finding of the equation
\begin{align}
 32 \alpha (2 + \beta) &(-\alpha^2 + 
    2 \alpha (1 + \beta) + (1 + \beta) (\beta + 4 s)) \notag\\
&+ (-\alpha^2 + 
    2 \alpha (5 + \beta) + (5 + \beta) (4 + \beta + 4 s)) (32 \alpha (1 + 
       \beta) + (-\alpha^2 + 
       2 \alpha (1 + \beta) + (1 + \beta) (\beta + 4 s))\notag\\
&\times (-\alpha^2 + 
       2 \alpha (3 + \beta) + (3 + \beta) (2 + \beta + 4 s)))=0
\end{align}
where $\alpha=d\sqrt{U_0},~\beta=(11+4s-d\sqrt{U_0})/2$ and $d^2U_0>(11+4s)^2.$ The critical polynomial reads in this case as
\begin{align*}
\mathpzc{P}_{3}^2(\alpha,\beta;s)&=
32 \alpha (1 + \beta) (-\alpha^2 + 
    2 \alpha (1 + \beta) + (1 + \beta) (\beta + 4 s) - 
    4 (5 + 2 \alpha + 2 \beta + 4 s))\notag\\
& + ( 
    2 \alpha (1 + \beta) -\alpha^2+ (1 + \beta) (\beta + 4 s)) (32 \alpha (2 + 
       \beta) + (-\alpha^2 + 
       2 \alpha (1 + \beta) + (1 + \beta) (\beta + 4 s) - 
       2 (3 + 2 \alpha + 2 \beta + 4 s)) \notag\\
&\times(-\alpha^2 + 
       2 \alpha (1 + \beta) + (1 + \beta) (\beta + 4 s) - 
       4 (5 + 2 \alpha + 2 \beta + 4 s)))
\end{align*}
\end{itemize}
\section{Conclusions}
\noindent We have given the sufficient and necessary conditions that ensure the existence of the polynomial solutions of a general confluent Heun equation. It is indeed exciting to know that these polynomial solutions are the source of finite sequences of orthogonal polynomials. The number of the orthogonal polynomials in each sequence equal to the degree of the polynomial solution of the confluent Heun equation. Although we discussed several properties of these orthogonal polynomials, a general theory that explain the zeros and evaluation of the discrete weight function will be of great interest. We hope the present study will stimulate further research on the theory of finite orthogonal polynomials.

\section{Acknowledgments}
\medskip
\noindent Partial financial support of this work under Grant No. GP249507 from the 
Natural Sciences and Engineering Research Council of Canada
 is gratefully acknowledged. 

\end{document}